\documentclass[10pt, conference, letterpaper]{IEEEtran}
\IEEEoverridecommandlockouts

\ifCLASSOPTIONcompsoc
  \usepackage[nocompress]{cite}
\else
  \usepackage{cite}
\fi

\usepackage{caption}
\usepackage[none]{hyphenat}
\usepackage{algpseudocode}
\usepackage{amsmath}
\usepackage{amsfonts}
\usepackage{amssymb}
\usepackage{subfigure}

\usepackage{graphicx}
\usepackage[table]{xcolor}
\usepackage{array,multirow}

\usepackage{multicol}

\usepackage{amsfonts}
\usepackage{booktabs}
\usepackage{siunitx}

\usepackage{color}

\usepackage{tikz}

\usepackage{enumitem}

\usepackage{url}
\makeatletter
\def\url@leostyle{%
  \@ifundefined{selectfont}{\def\UrlFont{\sf}}{\def\UrlFont{\small\ttfamily}}}
\makeatother
\newcommand{\eat}[1]{}
\usepackage[linesnumbered,ruled]{algorithm2e}
\usepackage{mdframed} 

\renewcommand{\IEEEbibitemsep}{1pt plus 2pt}
\makeatletter
\IEEEtriggercmd{\reset@font\normalfont\small}
\makeatother
\IEEEtriggeratref{1}

\usepackage{xcolor}
\definecolor{light-gray}{gray}{0.9}

\newenvironment{packed_enum}{%
  \begin{enumerate}%
  }{\end{enumerate}}

\usepackage{amsthm}

\newtheorem{lemma}{Lemma}

\newcolumntype{L}[1]{>{\raggedright\let\newline\\\arraybackslash\hspace{0pt}}m{#1}}

\definecolor{codegreen}{rgb}{0,0.6,0}
\definecolor{codegray}{rgb}{0.5,0.5,0.5}
\definecolor{codepurple}{rgb}{0.58,0,0.82}
\definecolor{backcolour}{rgb}{0.95,0.95,0.92}

\usepackage{arydshln}

\usepackage{listings}

\usepackage{listings, xcolor}

\definecolor{verylightgray}{rgb}{.97,.97,.97}

\lstdefinelanguage{Solidity}{
	keywords=[1]{anonymous, assembly, assert, balance, break, call, callcode, case, catch, class, constant, continue, constructor, contract, debugger, default, delegatecall, delete, do, else, emit, event, experimental, export, external, false, finally, for, function, gas, if, implements, import, in, indexed, instanceof, interface, internal, is, length, library, log0, log1, log2, log3, log4, memory, modifier, new, payable, pragma, private, protected, public, pure, push, require, return, returns, revert, selfdestruct, send, solidity, storage, struct, suicide, super, switch, then, this, throw, transfer, true, try, typeof, using, value, view, while, with, addmod, ecrecover, keccak256, mulmod, ripemd160, sha256, sha3}, 
	keywordstyle=[1]\color{blue}\bfseries,
	keywords=[2]{address, bool, byte, bytes, bytes1, bytes2, bytes3, bytes4, bytes5, bytes6, bytes7, bytes8, bytes9, bytes10, bytes11, bytes12, bytes13, bytes14, bytes15, bytes16, bytes17, bytes18, bytes19, bytes20, bytes21, bytes22, bytes23, bytes24, bytes25, bytes26, bytes27, bytes28, bytes29, bytes30, bytes31, bytes32, enum, int, int8, int16, int24, int32, int40, int48, int56, int64, int72, int80, int88, int96, int104, int112, int120, int128, int136, int144, int152, int160, int168, int176, int184, int192, int200, int208, int216, int224, int232, int240, int248, int256, mapping, string, uint, uint8, uint16, uint24, uint32, uint40, uint48, uint56, uint64, uint72, uint80, uint88, uint96, uint104, uint112, uint120, uint128, uint136, uint144, uint152, uint160, uint168, uint176, uint184, uint192, uint200, uint208, uint216, uint224, uint232, uint240, uint248, uint256, var, void, ether, finney, szabo, wei, days, hours, minutes, seconds, weeks, years},	
	keywordstyle=[2]\color{teal}\bfseries,
	keywords=[3]{block, blockhash, coinbase, difficulty, gaslimit, number, timestamp, msg, data, gas, sender, sig, value, now, tx, gasprice, origin},	
	keywordstyle=[3]\color{violet}\bfseries,
	identifierstyle=\color{black},
	sensitive=false,
	comment=[l]{//},
	morecomment=[s]{/*}{*/},
	commentstyle=\color{gray}\ttfamily,
	stringstyle=\color{red}\ttfamily,
	morestring=[b]',
	morestring=[b]"
}

\begin{document}

\title{EventWarden: A Decentralized Event-driven Proxy Service for Outsourcing Arbitrary Transactions in Ethereum-like Blockchains}

\author{

\IEEEauthorblockN{Chao Li}
\IEEEauthorblockA{Beijing Key Laboratory of Security and\\
Privacy in Intelligent Transportation\\
Beijing Jiaotong University\\
Beijing, China\\
Email: li.chao@bjtu.edu.cn}

\and

\IEEEauthorblockN{Balaji Palanisamy}
\IEEEauthorblockA{School of Computing and Information\\
University of Pittsburgh\\
Pittsburgh, USA\\
Email: bpalan@pitt.edu}

}

\maketitle


\begin{abstract}

Transactions represent a fundamental component in blockchains as they are the primary means for users to change the blockchain state. Current blockchain systems such as Bitcoin and Ethereum require users to constantly observe the state changes of interest or the events taking place in a blockchain and requires the user to explicitly release the required transactions to respond to the observed events in the blockchain.
This paper proposes \texttt{EventWarden}, a decentralized event-driven proxy service for users to outsource transactions in Ethereum-like blockchains.
\texttt{EventWarden} employs a novel combination of smart contracts and blockchain logs.
\texttt{EventWarden} allows a user to create a proxy smart contract that specifies an interested event and also reserves an arbitrary transaction to release.
Upon observing the occurrence of the prescribed event, anyone in the Blockchain network can call the proxy contract to earn the service fee reserved in the contract by proving to the contract that the event has been recorded into blockchain logs, which then automatically triggers the proxy contract to release the reserved transaction.
We show that the reserved transaction can only get released from the proxy contract when the prescribed event has taken place.
We also demonstrate that as long as a single member in the Blockchain network is incentivized by the service fee to call the proxy contract after the prescribed event has taken place, the reserved transaction is guaranteed to get released.
We implement \texttt{EventWarden} over the Ethereum official test network. The results demonstrate that \texttt{EventWarden} is effective and is ready-to-use in practice.

\end{abstract}

\section{Introduction}

Blockchains such as Bitcoin\cite{nakamoto2008bitcoin} and Ethereum\cite{buterin2014next} are ledgers of transactions performed by nodes in blockchain networks on a global state.
Transactions form the fundamental component of blockchains as they are the primary means for users to change the blockchain state.
For example, transactions allow users to transfer funds among each other in Bitcoin-like blockchains and in Ethereum-like blockchains that support smart contracts~\cite{wood2014ethereum}, transactions enable users to create new smart contracts and invoke functions within existing smart contracts.
Even though many scenarios require users to release a transaction,
it is very common to see that users need to release transactions for responding to prior state changes in the blockchain in a timely manner.
Thus, current approaches require users to constantly keep monitoring their interested events.
\underline{E}vent-driven \underline{T}ransaction (\textit{ET}) refers to a class of service that enables a user to outsource a
transaction to get executed to change the blockchain state immediately after a certain state change (or event) has taken place in the blockchain.
Many scenarios require event-driven transactions in practice. 
For example, in a smart-contract-powered game such as the popular CryptoKitties~\cite{cryptokitties} in Ethereum, Alice may want to purchase a digital kitty with a rare color (i.e., via a transaction) once the kitty appears in the market (i.e., an event) while Alice may not be capable of watching the market for days. 
She may use the event-driven transaction by outsourcing the transaction to someone else and require the transaction to get released to purchase the kitty immediately after the kitty becomes available.
In another example, if Bob attends a smart-contract-powered auction~\cite{SealedBidAuction} and would like to increase his bid (i.e., via a transaction) once a bid higher than his old bid appears (i.e., an event), Bob may outsource the transaction and require the transaction to get executed only after the event occurs.
Besides the above-mentioned use cases of event-driven transactions, recent works on smart-contract-powered applications and protocols ranging from on-chain designs to side-chain and off-chain designs are heavily employing an emerging design pattern called \textit{challenge-response}~\cite{das2019fastkitten,mccorry2019pisa,poon2017plasma,kalodner2018arbitrum}, which can also benefit from adopting event-driven transactions.
Specifically, in the design of a multi-party smart-contract-powered protocol, a naive design pattern is to first verify the correctness of an action performed by a participant and then change the blockchain state based on the verified action.
However, the verification of actions may be quite expensive in Ethereum-like blockchains and may become even impossible in some cases.
In contrast, in the \textit{challenge-response} design pattern, each action is assigned a time period called the challenge period, during which other participants may choose to challenge the correctness of the action with counterexamples. Participants performing incorrect actions will then be found out and their security deposits will be confiscated.
Therefore, the \textit{challenge-response} design pattern incentivizes all participants to stay honest, which significantly reduces the cost of implementing a protocol.
It is easy to see that the timely responses for challenging incorrect actions are vital to such designs, otherwise the designs become insecure.
Event-driven transactions can thus facilitate the emerging \textit{challenge-response} design pattern by outsourcing transactions for challenging events such as disputes or suspicious actions.

Most of the current implementations (e.g., BlueOrion~\cite{blueorion} and Oraclize~\cite{oraclize}) of event-driven transactions (\textit{ET}) are heavily centralized.
These services require the users to entirely trust the centralized servers and their security properties are solely limited to a single point of trust.
More importantly, even in scenarios when the service providers are considered trustworthy, the services are still prone to unpredictable security breaches or insider attacks that are beyond the control of the service providers~\cite{chen2018certchain,hu2018searching}.
On the other hand, the emergence of Blockchain technologies such as Ethereum\cite{buterin2014next} and smart contracts\cite{wood2014ethereum} provides significant potential for new security designs that support a decentralized implementation of \textit{ET} to overcome the single point of trust issues associated with centralized approaches.
Recently, a few works~\cite{Clock,ning2018keeping,li2018decentralized-2,Kimono,SilentDelivery,mccorry2019pisa} have focused on such decentralized designs in Ethereum.
Nevertheless, their designs only support a certain type of event (e.g., reach of deadlines, disputes in state channels) or a single type of transaction in Ethereum (i.e., function invocation transaction).

In this paper, we propose \texttt{EventWarden}, a decentralized event-driven proxy service for users to outsource any type of transaction in Ethereum-like blockchains.
\texttt{EventWarden} employs a novel combination of smart contracts and blockchain logs.
\texttt{EventWarden} allows a user to create a proxy smart contract that specifies an interested event and also reserves an arbitrary transaction to release.
Upon observing the occurrence of the prescribed event, anyone in the Blockchain network can call the proxy contract to earn the service fee reserved in the contract by proving to the contract that the event has been recorded into blockchain logs, which then automatically triggers the proxy contract to release the reserved transaction.
We show that the reserved transaction can only get released from the proxy contract when the prescribed event has taken place.
We also demonstrate that as long as a single member in the Blockchain network is incentivized by the service fee to call the proxy contract after the prescribed event has taken place, the reserved transaction is guaranteed to get released.

In summary, this paper makes the following key contributions:
\begin{itemize}[leftmargin=*]
\item To the best of our knowledge, \texttt{EventWarden} is the \textit{first} decentralized proxy service designed for the general use of event-driven transactions in Ethereum-like blockchains.
\item After the service has been set up, \texttt{EventWarden} completely isolates the service execution from the state of users, without requiring any assistance from the user side.
\item We emphasize that \texttt{EventWarden} is a general approach that supports \textit{all} types of transaction in Ethereum, including fund transfer transaction, contract creation transaction and function invocation transaction.
\item \texttt{EventWarden} also supports \textit{all} types of events in Ethereum as long as the corresponding smart contracts write events into the blockchain logs.
\item We implement \texttt{EventWarden} over the Ethereum official test network. 
The results demonstrate that \texttt{EventWarden} is effective and is straight-forward to be used in practice.
\end{itemize}

The rest of this paper is organized as follows: 
We discuss related work in Section~\ref{s2} and introduce preliminaries in Section~\ref{s3}.
In Section~\ref{s4}, we propose the architecture designed for \texttt{EventWarden}.
Then, in Section~\ref{s5}, we propose the protocol designed for \texttt{EventWarden}.
We present the security analysis of \texttt{EventWarden} in Section~\ref{s6} and the implementation and evaluation of \texttt{EventWarden} in Section~\ref{s7}.
Finally, we conclude in Section~\ref{s8}.

\section{Related work}
\label{s2}

In this section, we briefly review related studies and discuss relevant techniques, which can be roughly divided into four categories.

\subsection{Native mechanisms}

Bitcoin was designed with a native mechanism named \textit{Timelocks}~\cite{antonopoulos2017mastering}.
Each transaction in Bitcoin can be set with a \textit{nLocktime} and the transaction can only be accepted by the network after the time point indicated by \textit{nLocktime}.
Besides, each Bitcoin transaction may involve one or multiple Unspent Transaction Outputs (UTXOs) and each UTXO may include an UTXO-level timelock named \textit{Check Lock Time Verify} (CLTV) that makes the UTXO available only after the specified time point. We consider this \textit{Timelocks} mechanism a good way to schedule fund transfer transactions and events relative to time, but it is difficult to be generalized to support other types of transaction in Ethereum-like blockchains beyond Bitcoin or other types of events.

\subsection{Client-side tools}

There are many tools at the client side that can achieve event-driven transactions.
For example, \textit{parity}~\cite{parity}, a popular Ethereum client, allows its users to prescribe a time point that they would like a transaction to be sent into the Ethereum network by the client.
However, since the scheduled transaction is locally stored at users' machines before the prescribed time point, the usage of such tools demands users' machines to keep connecting with the blockchain network, which fails to isolate the service execution with the state of users.

\subsection{Centralized services}

Oraclize~\cite{oraclize} is a blockchain oracle service that takes the role of a trusted third party (TTP) to execute a pre-scheduled transaction on behalf of a user at a future time point. 
Similarly, BlueOrion~\cite{blueorion} enables date related payments on the Stellar Ecosystem~\cite{stellar}, an Ethereum-like blockchain.
The limitations of these centralized services include both a single point of trust and a single point of control.

\subsection{Decentralized services}

A recent project called \textit{Ethereum Alarm Clock}~\cite{Clock} allows a user to deploy a request contract to the Ethereum network with a future time-frame as well as a reward and if any account is interested in the reward, the account can invoke the request contract during the prescribed time-frame to earn the reward by making the function invocation transaction maintained by the request contract get executed.
However, this scheme supports only a single type of transaction and a single type of event, 
namely the arrival of prescribed time-frame.
A more recent work~\cite{li2018decentralized-2} further employs threshold secret sharing~\cite{shamir1979share} to offer privacy protection for scheduling function invocation transactions that involve sensitive arguments (e.g., bid, vote).
Nevertheless, similar to \textit{Ethereum Alarm Clock}, this work only supports function invocation transaction and events relative to time.
Another recent work named PISA~\cite{mccorry2019pisa} enables parties in state channels to delegate to a third
party, called the custodian, to cancel execution forks on their behalf.
However, this work only supports events relative to state channels. 

In summary, our work in this paper tackles the key limitations of the state-of-the-art decentralized approaches using smart contracts~\cite{Clock,li2018decentralized-2,mccorry2019pisa}. 
That is, by leveraging proxy contracts and blockchain logs, our scheme supports \textit{all} types of transaction (fund transfer, function invocation and contract creation) in Ethereum, as well as \textit{all} types of events in Ethereum as long as the corresponding smart contracts write events into the blockchain logs.
To the best of our knowledge, \texttt{EventWarden} is the \textit{first} decentralized proxy service designed for the general use of event-driven transactions in Ethereum-like blockchains.

\section{Preliminaries}
\label{s3}
In this section, we discuss the preliminaries about smart contracts and Ethereum.
While we discuss smart contracts in the context of Ethereum~\cite{wood2014ethereum}, we note that our solutions are also applicable to a wide range of other Ethereum-like blockchains.

\subsection{Account types}
There are two types of accounts in Ethereum, namely External Owned Accounts (EOAs) and Contract Accounts (CAs).
To interact with the Ethereum blockchain, a user needs to own an EOA by locally creating a pair of keys.
Specifically, the public key $pk_{EOA}$ can generate a 20-byte address $addr(EOA)$ to uniquely identify the EOA and the private key $sk_{EOA}$ can be used by the user to sign transactions or other types of data.
Then, any user can create a smart contract by sending out a contract creation transaction from a controlled EOA.
The 20-byte address $addr(CA)$ of the created smart contract is generated in a deterministic and predictable way and becomes the unique identity of the contract account.

\subsection{Transactions and messages}
\label{s3.2}
The state of Ethereum blockchain can only be changed by the external world (i.e., EOAs) using transactions.
A transaction is a serialized binary message sent from an EOA that contains the following key elements:
\begin{itemize}
  \item \textit{recipient}: the recipient account address;
  \item \textit{value}: the amount of ether\footnote{\begin{scriptsize} The native cryptocurrency in Ethereum, denoted by $\mathsf{\Xi}$. \end{scriptsize}} to send to the recipient;
  \item \textit{data}: the binary data payload;
\end{itemize}

Depending on the type of \textit{recipient}, transactions can be divided into three categories.

\noindent \textbf{Fund transfer transaction}:
A transaction with an EOA as \textit{recipient} and a non-empty \textit{value} is a fund transfer transaction, which is used to transfer an amount of ether from the EOA creating the transaction to the \textit{recipient} EOA.

\noindent \textbf{Function invocation transaction}: 
When a transaction involves an CA as \textit{recipient} as well as a non-empty \textit{data}, it is usually a function invocation transaction for calling a function within an existing smart contract.

\noindent \textbf{Contract creation transaction}: 
In Ethereum, there is a special type of transaction for creating new smart contracts.
Such a transaction, usually called a contract creation transaction, carries a special \textit{recipient} address 0x0, an empty \textit{value} and a non-empty \textit{data} payload.
A smart contract (or contract) in Ethereum is a piece of program created using a high-level contract-oriented programming language such as \textit{Solidity}~\cite{Solidity2017}.
After compiling into a low-level bytecode language called Ethereum Virtual Machine (EVM) code, the created contract is filled into a contract creation transaction as the \textit{data} payload.

To make the transaction get executed to change the state of Ethereum blockchain, the transaction should be broadcasted to the entire Ethereum network formed by tens of thousands of miner nodes.
Following the Proof-of-Work (PoW) consensus protocol~\cite{nakamoto2008bitcoin}, all the miners in Ethereum competitively solve a blockchain puzzle and the winner packages the received transactions into a block and appends the new block to the end of Ethereum blockchain.
From then on, it is hard to tamper with the blockchain state updated by the transaction (i.e., transferred fund, executed function or created contract) as each miner maintains a copy of the new block and is aware of changes made by transactions within the new block and an adversary has to falsify majority of these copies in order to change the network consensus about the glocal state. 

There is another important concept in Ethereum that is highly relevant to transactions (or \textit{tx}), namely the messages (or \textit{msg}).
In Ethereum, smart contracts can be programmed in a way that enables EVM-level opcodes for system operations, such as \texttt{CREATE} and \texttt{CALL}, through which contracts are able to perform advanced operations such as creating new contracts or calling other contracts by sending out messages to change the blockchain state.
It is worth noting that CAs are not able to change the blockchain state by themselves.
Instead, any message creation at a CA must be triggered by a transaction sent from EOAs or another message created at other CAs, which is also triggered by EOAs from the root.
This property offers two approaches to update the state of a smart contract (say, CA1) in Ethereum:
(1) send a function invocation transaction ($tx$) to CA1, which directly updates the state of CA1;
(2) send a $tx$ to another contract CA2, where the $tx$ triggers CA2 to send a message $msg$ to CA1 and update the state of CA1.
In fact, our study found that the second approach supports not only function invocation transactions, but also the other two types of transactions in Ethereum.


   

\subsection{Gas system}
\label{s3.3}
In order to either deploy a new contract or call a deployed contract in Ethereum, one needs to spend Gas.
Based on the complexity of the contract or that of the called function, an amount of ether needs to be spent in order to purchase an amount of Gas, which is then paid to the miner that creates the new block. 

\begin{figure*}
\centering
{
   
    \includegraphics[width=13.5cm,height=7cm]{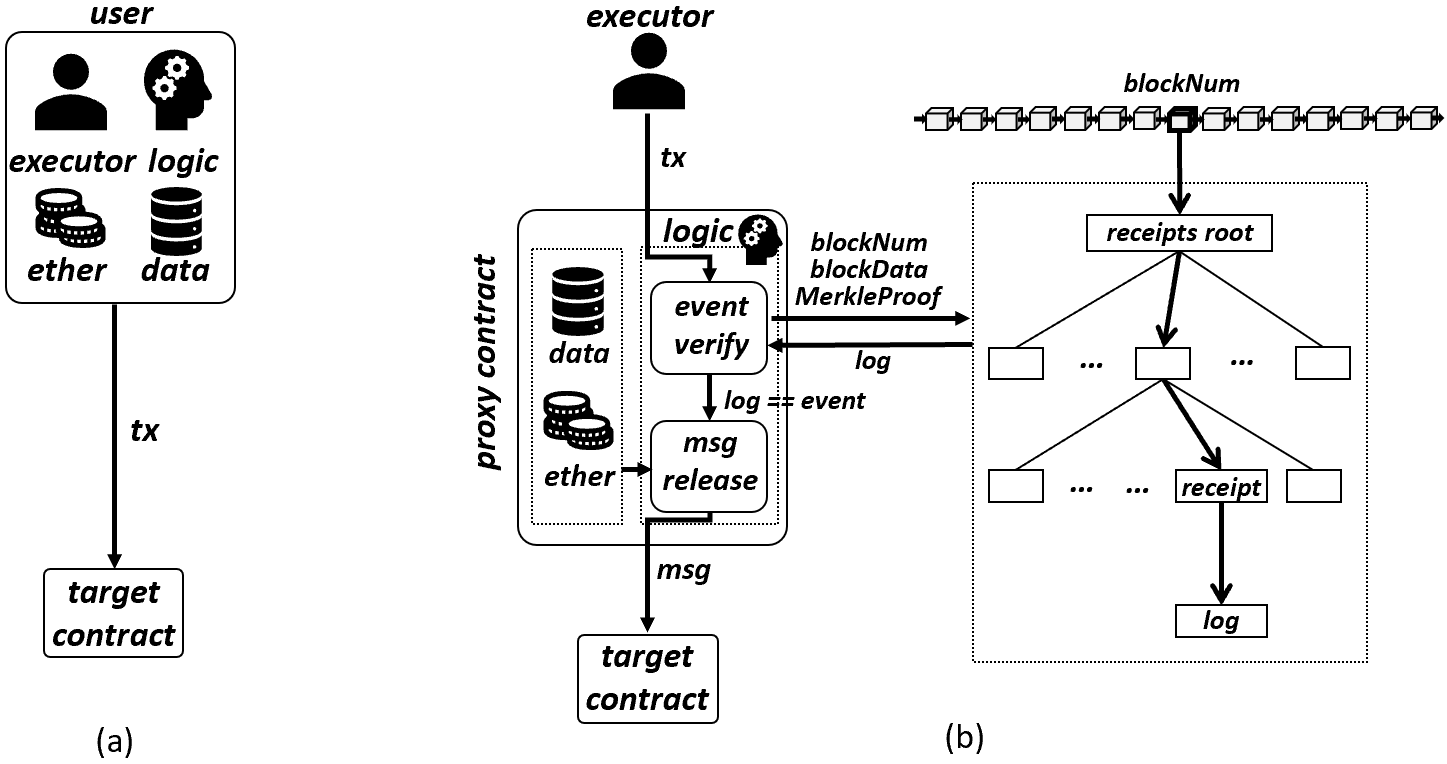}
}
\vspace{-1mm}
\caption {The architectures for Event-driven Transaction.
(a) the user-driven architecture (b) the proxy-contract-driven architecture. }
\vspace{-5mm}
\label{architecture}
\end{figure*}

\subsection{Events and logs}
\label{s3.4}
Events are inheritable members of smart contracts and are used for emphasizing the effect of each transaction~\cite{wood2014ethereum}. When an event is called along with a function invocation transaction, the LOG opcodes in Ethereum virtual machine store the event in the transaction’s log, as a part of the transaction receipt.
These logs are associated with the address of the contract and will be incorporated into the blockchain and stay there as long as a block is accessible.
Meanwhile, each block maintains a Merkle Patricia tree for receipts of all transactions within the block and the root of this tree, denoted as the \textit{receiptsRoot}, is stored in the block header.
Therefore, a contract similar to CryptoKitties~\cite{cryptokitties} may be designed in a way such that each function invocation transaction for releasing a kitty to the market notifies the appearance of the kitty with an event describing the properties of the kitty (e.g., $event(kitty\_color)$).
Similarly, an auction contract may employ events to notify the update of the highest bid (e.g., $event(updated)$).
With events, users are capable of tracking the state changes of interested smart contracts in a very efficient way.
Moreover, the facts that events are recorded in transaction receipts as logs and the Merkle root of transaction receipts are written into block header imply that the existence of events is provable using \textit{receiptsRoot} and the reliability of the Merkle proof is endorsed by the reliability of blockchain.
Motivated by these findings, we employ events and logs in the design of \texttt{EventWarden}.

In the next two sections, we present the architecture and protocol designed for \texttt{EventWarden}, respectively.
We summarize the notations that will be used in this section and in the rest of this paper in TABLE~\ref{t1}.

\begin{table}
\begin{center}
\begin{tabular}{|c|p{6.5cm}|}
\hline
\textbf{notation} & \textbf{description} \\
\hline
$U$ & a user of Event-driven Transaction (ET) \\
$E$ & an executor in ET \\
$C$ & a smart contract \\
$C.fun()$ & function $fun()$ within contract $C$ \\
$\Rightarrow$ & invoke a function within a contract\\
$addr(*)$ & an address of an EOA or a CA\\
\hline
\end{tabular}
\end{center}
\caption{Summary of notations.}      
\vspace{-5mm}
\label{t1}
\end{table}

\section{EventWarden: Architecture}
\label{s4}
In this section, we first present the key components of the architecture for event-driven transactions. We then present a naive user-driven architecture and finally introduce the proxy-contract-driven architecture designed for \texttt{EventWarden}.

\subsection{The four key components}

An architecture for event-driven transaction (\textit{ET}) includes four key components.
First, it requires a storage place for maintaining the elements of a scheduled transaction, including \textit{recipient}, \textit{value} and \textit{data} presented in Section~\ref{s3.2}.
Second, it requires an amount of ether not less than the amount indicated by the \textit{value} element, which is guaranteed to be available when the event occurs.
Third, it requires an EOA to execute the broadcasting of scheduled transaction after the event to change the blockchain state as expected.
Finally, it requires the logic to indicate the occurrence of the event.
We summarize the four components as \textit{et-data}, \textit{et-ether}, \textit{et-executor} and \textit{et-logic}.

Next, we present two different architectures for \textit{ET} and discuss how they handle the four key components, respectively.
We start by introducing the user-driven architecture shown in Fig.~\ref{architecture}.(a), where \textit{ET} is simply performed by a user who needs to handle all these components by herself and also needs to keep connecting with the blockchain network.
We then introduce the proxy-contract-driven architecture in Fig.~\ref{architecture}.(b), where a proxy contract is employed to manage \textit{et-data}, \textit{et-ether} and \textit{et-logic} and an EOA is recruited to take the role of \textit{et-executor} so that the user is totally relieved from all the burdens of handling components and keeping on-line.

\subsection{The user-driven architecture}

As introduced in Section~\ref{s3.2}, a straight-forward approach of implementing event-driven transactions would be to ask the user to store the elements of the scheduled transaction at her local machine and make the transaction get broadcasted via a client-side tool such as \textit{parity}~\cite{parity} once an event occurs.
In this scenario, user stores \textit{et-data} at a local machine (i.e., a PC) with the local programs handling \textit{et-logic} and employs a controlled EOA to both provide \textit{et-ether} and serve as \textit{et-executor}.
Then, upon finding an occurrence of the prescribed event, the client-side tool pushes the scheduled transaction into the blockchain network from the local machine and the amount of ether indicated by the \textit{value} element is transferred from the controlled EOA to the address indicated by the \textit{recipient} element.
In summary, this user-driven architecture is easy to be implemented.
However, even though this approach allows a user to be physically absent during the occurrence of events, it still demands the local machine that maintains \textit{et-data} to be connected with the blockchain network, which makes user-driven architecture hard to be adopted as a general approach that can isolates the service execution completely from the user side after the service has been set up.

\subsection{The proxy-contract-driven architecture}
\label{s4.3}

To isolate the service execution from the user side, all the four components need to be migrated from the user.
A naive way of realizing this goal would be to ask the user to outsource all the components to an EOA recruited from the Ethereum network so that the recruited EOA would be able to complete \textit{ET} on behalf of the user.
However, there are three severe consequences that are likely to occur:
(1) the recruited EOA may falsify \textit{et-data};
(2) the recruited EOA may embezzle received \textit{et-ether};
(3) the recruited EOA may violate the indication of \textit{et-logic}.
All these undesirable execution results can hardly be prevented in this naive approach. Hence, a natural question is that, can we design an architecture for \textit{ET} that offers all-or-nothing execution results?
That is, we would like either the scheduled transaction to be correctly executed along with the occurrence of the prescribed event, or the scheduled transaction to be never executed and the \textit{et-ether} to be transferred back to the user.
To achieve this goal, we recognize that \textit{et-data}, \textit{et-ether} and \textit{et-logic} need to be handled by a trusted party and executed in a deterministic manner, so we design a proxy smart contract ($C_{proxy}$) that can act on behalf of a user to manage the three components in a decentralized, trustworthy and deterministic way.

\vspace{-3mm}
\begin{algorithm}
    \footnotesize
    \SetKwInOut{Input}{Input}
    \SetKwInOut{Output}{Output}

    \Input{$blockNum, blockData, MerkleProof, et$-$event$, \\$et$-$data, et$-$ether$.}
    \Output{$msg$.}

    $receiptsRoot \gets verifyBlock(blockNum,blockData)$\;
    $receipt \gets verifyProof(receiptsRoot,MerkleProof)$\;
    $result \gets verifylog(receipt,et$-$event)$\;
    \If{$result == TRUE$} 
    {
      $msgRelease(et$-$data, et$-$ether)$\;
    }

    \caption{Proxy contract logic}
    \label{A1} 
\end{algorithm}
\vspace{-3mm}

As illustrated in Algorithm~\ref{A1}, with the proxy contract, \textit{ET} can get completed in two steps:
\begin{itemize}[leftmargin=*]
\item \textit{Event verify}:  
Immediately after the event occurs, any EOA in Ethereum is capable of calling the proxy contract $C_{proxy}$ with a function invocation transaction taking three arguments, namely the number of block carrying the event as log inside a transaction receipt ($blockNum$), the data of that block ($blockData$) and the Merkle proof for proving the existence of the receipt ($MerkleProof$).
With these arguments, the proxy contract is capable of first verifying the correctness of $blockData$ according to $blockNum$ and fetching the $receiptsRoot$ from $blockData$ (line 1), then verifying the existence of the declared $receipt$ using $MerkleProof$ and $receiptsRoot$ (line 2) and finally verifying the existence of the declared log and checking whether the log is equivalent to the event prescribed by user, namely \textit{et-event} (line 3).
\item \textit{Message release}:
In case that the log correctly indicates the occurrence of the prescribed event, the proxy contract releases the reserved transaction by sending out a message using \textit{et-data} and \textit{et-ether}.
\end{itemize}

The proxy-contract-driven architecture can effectively prevent the three undesirable consequences:
(1) the data recorded in  $C_{proxy}$, namely \textit{et-data}, can only get compromised by attacking the Ethereum blockchain;
(2) the amount of ether in $C_{proxy}$, namely \textit{et-ether}, is only allowed to be either transferred to the address of \textit{recipient} or to the user using a native \textit{selfdestruct} function, so no one would be able to embezzle \textit{et-ether};
(3) the correctness of implementation of \textit{et-logic} is provable, which we leave to Section~\ref{s6} to present in detail.
Next, we present the protocol designed for \texttt{EventWarden}.

\section{EventWarden: Protocol}
\label{s5}

In this section, we start by describing the event-driven transaction (ET) as a two-phase process in the context of the proposed proxy-contract-driven architecture.
We then present the protocol designed for \texttt{EventWarden}.
Throughout Section~\ref{s5}, we assume that the user is scheduling a function invocation transaction to call a function within a target smart contract denoted as $C_{target}$.
While we present the protocols in the context of a single type of transaction, we note that our protocols are also applicable to the other two types.

\subsection{ET as a two-phase process}
We describe the \textit{ET} problem as a two-phase process in the context of the proposed proxy-contract-driven architecture:

\begin{itemize}[leftmargin=*]

\item \textit{ET.schedule}: 
The user $U$ creates a proxy contract $C_{proxy}$ with a balance (if needed) and sends her service request to a service hub contract $C_{hub}$, which notifies the $C_{proxy}$ to potential executors $Es$.
\item \textit{ET.execute}:
Upon detecting the occurrence of the event specified in $C_{proxy}$, executor $E$ can invoke $C_{proxy}$ with a function invocation transaction to trigger the release of the reserved transaction in the form of a message.

\end{itemize}

\begin{figure}
\centering
{
   
    \includegraphics[width=8.4cm,height=5.5cm]{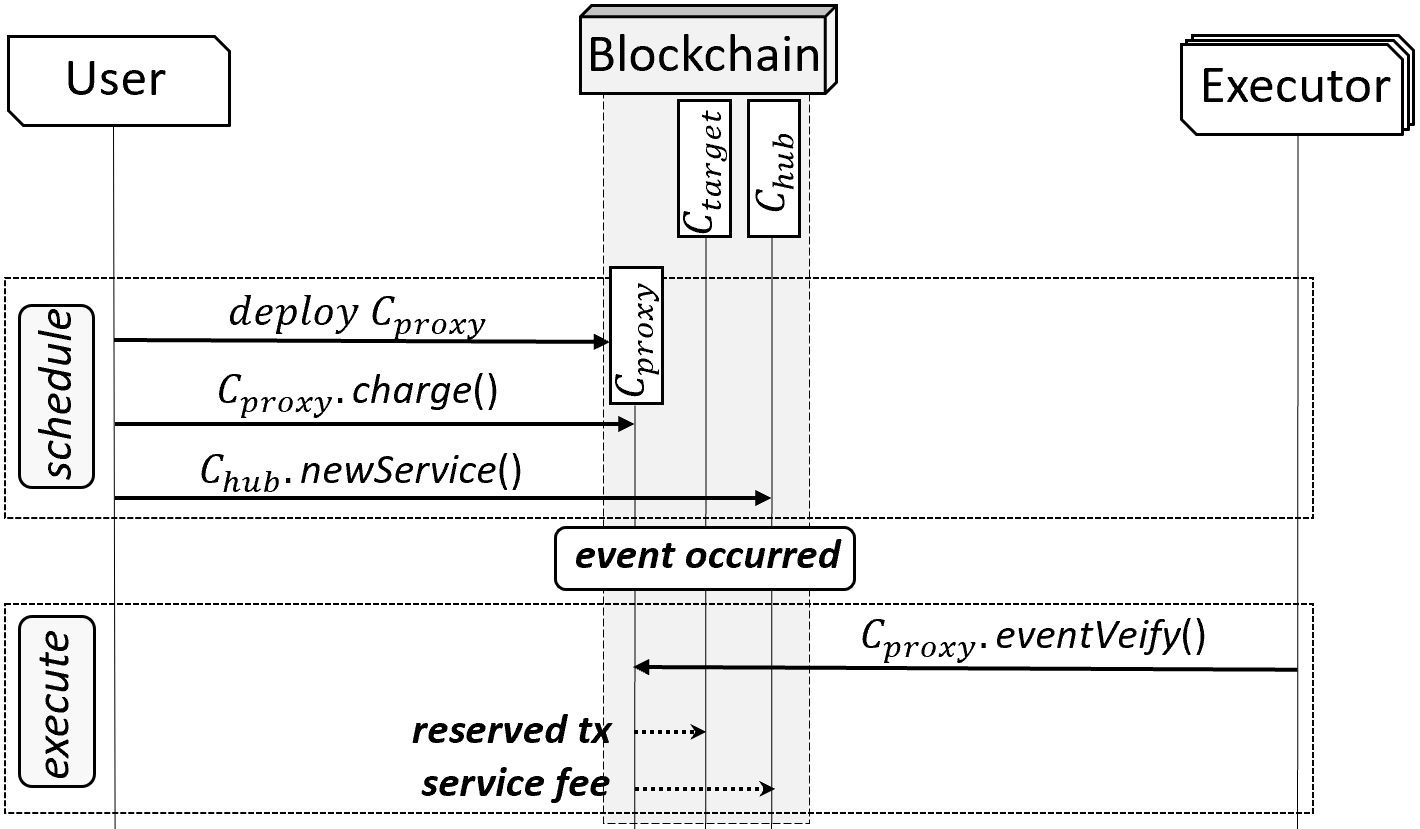}
}
\caption {Protocol sketch}
\vspace{-5mm}
\label{protocol_sketch_02} 
\end{figure}

\subsection{The protocol}

We now present the protocol in detail.
We sketch the protocol in Fig.~\ref{protocol_sketch_02} and present the formal description in Fig.~\ref{protocol_detail_2}.
Concretely, the protocol employs a service hub contract ($C_{hub}$) to manage all $ET$ service requests, so executors only need to track events released from $C_{hub}$ to get notified about new $ET$ service requests.
We next describe the two-phase process in detail.

\noindent \textbf{\textit{ET.schedule}}: 
User $U$ first deploys a proxy contract $C_{proxy}$.
User $U$ then needs to send an amount of ether via function $charge()$ to $C_{proxy}$. The amount of ether received by $C_{proxy}$ should be the sum of service fee for paying the executor and \textit{et-ether} (if not empty). 
After that, user $U$ sets up a new \textit{ET} service with $C_{hub}$ via function $newService()$ and specifies the service details, namely the address (i.e., $addr(*)$) of the proxy contract $C_{proxy}$.
Upon getting notified by $C_{hub}$ about this new service request, executors $Es$ can start monitoring the occurrence of the event specified in the corresponding $C_{proxy}$.

\noindent \textbf{\textit{ET.execute}}:
Upon detecting the occurrence of the event specified in $C_{proxy}$, any executor $E$ is capable of constructing $blockNum$, $blockData$ and $MerkleProof$ locally and call function $eventVerify()$ inside $C_{proxy}$ with a function invocation transaction.
Upon getting invoked, function $eventVerify()$ follows Algorithm~\ref{A1} to verify the occurrence of the specified event.
Function $eventVerify()$ then call a private $msgRelease()$ function to release the prescribed message and then transfer the service fee to the executor.

\section{Security analysis}
\label{s6}

In this section, we present the security analysis for the proposed \texttt{EventWarden}. 

\begin{lemma}
\label{lemma1}
The reserved transaction can only get released from the proxy contract when the prescribed event has taken place.
\end{lemma}

\begin{proof}
As illustrated in Algorithm~\ref{A1}, upon getting invoked by a function invocation transaction, the logic inside the proxy contract, namely the $eventVerify()$ function would verify the occurrence of the prescribed event before releasing the reserved transaction. Concretely, verifying the occurrence of the prescribed event is equivalent to verifying the existence of the log for storing the event in the blockchain. 
As presented in Section~\ref{s3.4}, the reliability of the verification is endorsed by the reliability of blockchain. 
Together, as long as we could trust the blockchain, we could also trust the correctness of the verification done in the proxy contract.
\end{proof}

\begin{figure}
\begin{minipage}{0.5\textwidth}
\begin{small}
\begin{mdframed}[innerleftmargin=8pt]

\noindent \textbf{ET.schedule:} 
\begin{packed_enum}[leftmargin=*]
  \item User $U$ deploys a proxy contract $C_{proxy}$.
  \item User $U \Rightarrow C_{proxy}$: $charge(ether)$.
  \item User $U \Rightarrow C_{hub}$: $newService(addr(C_{proxy}))$.
\end{packed_enum}

\noindent \textbf{ET.execute:} 
\begin{packed_enum}[leftmargin=*]
  \setcounter{enumi}{3}
  \item Any executor $E \Rightarrow C_{proxy}.eventVerify(blockNum, \\blockData, MerkleProof)$.
\end{packed_enum}

\end{mdframed}
\end{small}
\end{minipage}

\captionof{figure}{
Formal protocol
}
\vspace{-5mm}
\label{protocol_detail_2}
\end{figure}

\begin{lemma}
\label{lemma2}
As long as a single member in Ethereum is incentivized by the service fee to call the proxy contract after the prescribed event has taken place, the reserved transaction is guaranteed to get released.
\end{lemma}

\begin{figure}
\begin{minipage}{\linewidth}
\begin{lstlisting}[
linewidth=8.7cm,
language=Solidity,
basicstyle=\scriptsize,
caption=The \textit{eventVerify()} function,
label={I1},
frame=single,
numbers=left,
numbersep=5pt,
breaklines=true,
% basicstyle=\ttfamily
  % postbreak=\mbox{\textcolor{red}{$\hookrightarrow$}\space},
breakatwhitespace=true
]
pragma solidity ^0.5.0;
import "RLPReader.sol";

contract C_proxy {
  ...
  using RLPReader for RLPReader.RLPItem;
  using RLPReader for bytes;
  function eventVerify(
        bytes memory _rlpBlockData, uint _blockNum, 
        bytes memory _rlpMerkleRoot, uint _idMerkleRoot,
        bytes memory _rlpMerkleBranch, uint _idMerkleBranch,
        bytes memory _rlpMerkleLeaf
        ) public {
    // event verify
    require(keccak256(_rlpBlockData) == 
        blockhash(_blockNum));
    RLPReader.RLPItem[] memory lsBlockHeader = 
        _rlpBlockData.toRlpItem().toList(); 
    bytes32 receiptsRoot = 
        bytes32(lsBlockHeader[5].toUint());
    require(keccak256(_rlpMerkleRoot) == receiptsRoot);
    RLPReader.RLPItem[] memory lsMerkleRoot = 
        _rlpMerkleRoot.toRlpItem().toList(); 
    bytes32 merkleBranch = 
        bytes32(lsMerkleRoot[_idMerkleRoot].toUint());
    require(keccak256(_rlpMerkleBranch) == merkleBranch);
    RLPReader.RLPItem[] memory lsMerkleBranch = 
        _rlpMerkleBranch.toRlpItem().toList(); 
    bytes32 merkleLeaf = 
        bytes32(lsMerkleBranch[_idMerkleBranch].toUint());
    require(keccak256(_rlpMerkleLeaf) == merkleLeaf);
    RLPReader.RLPItem[] memory lsMerkleLeaf = 
        _rlpMerkleLeaf.toRlpItem().toList(); 
    bytes memory rlpReceipt = lsMerkleLeaf[1].toBytes();
    RLPReader.RLPItem[] memory lsReceipt = 
        rlpReceipt.toRlpItem().toList(); 
    bytes memory rlpLog = lsReceipt[3].toBytes();
    require(keccak256(rlpLog) == specifiedEvent);
    // message release 
    msgRelease();
    msg.sender.transfer(service_fee);
  }
  ...
}
\end{lstlisting}
\end{minipage}
\end{figure}

\begin{figure}
\begin{minipage}{\linewidth}
\begin{lstlisting}[
linewidth=8.7cm,
language=Solidity,
basicstyle=\scriptsize,
caption=The \textit{msgRelease()} function,
label={I2},
frame=single,
numbers=left,
numbersep=5pt,
breaklines=true,
% basicstyle=\ttfamily
  % postbreak=\mbox{\textcolor{red}{$\hookrightarrow$}\space},
breakatwhitespace=true
]
contract C_proxy {
  ...
  address payable _recipient = ...;
  uint _value = ...;
  function fundTransfer() private {  
    _recipient.transfer(_value);
  }

  address _recipient = ...;
  string _function_selector = ...; 
  uint _arg1; 
  uint _arg2;
  function functionInvocation() private {  
    _recipient.call(abi.encodeWithSignature(
        _function_selector, _arg1, _arg2));
  }

  bytes _bytecode = ...;
  function contractCreation() private {  
    address _deployedAddr;
    bytes memory _bc = _bytecode;
    assembly {
      _deployedAddr := create(0, add(_bc, 0x20), mload(_bc))
    }
  }
  ...
}
\end{lstlisting}
\end{minipage}
\end{figure}

\begin{proof}
As we have discussed in Section~\ref{s4.3}, the proposed proxy-contract-driven architecture isolates \textit{et-executor} from \textit{et-data}, \textit{et-ether} and \textit{et-logic}.
Concretely, \texttt{EventWarden} outlines no eligibility requirements for becoming executors. Any EOA in Ethereum are eligible for providing services for \texttt{EventWarden}. 
This is because that the vital components, namely  \textit{et-data}, \textit{et-ether} and \textit{et-logic} are maintained by the trustworthy proxy contract, so the need for verifying the qualification of executors is minimized.
Thanks to this strategy, every member from the huge Ethereum community becomes a potential service provider for \texttt{EventWarden} and the solid underlying community could make \texttt{EventWarden} easier to get launched in practice.
Moreover, anyone in Ethereum is capable of constructing $blockNum$, $blockData$ and $MerkleProof$ locally and calling the $eventVerify()$ function to complete the service, so the service is completely decentralized. 
This fact also facilitates the reliability of the service because it is impracticable for anyone to intercept all function invocation transactions sent by different EOAs from different places in the world for calling $eventVerify()$.
\end{proof}

\section{Implementation and evaluation}
\label{s7}
In this section, we present the implementation and evaluation for \texttt{EventWarden} in detail.

\subsection{Implementation}
We programmed both the proxy contract and hub contract using the contract-oriented
programming language \textit{Solidity}~\cite{Solidity2017} and we tested the contracts over the
Ethereum official test network \textit{Rinkeby}~\cite{Rinkeby2017}.
We employed the \textit{solidity-rlp} library~\cite{solidity-rlp} for decoding data encoded with the Ethereum RLP (Recursive Length Prefix) rules~\cite{wood2014ethereum} in smart contracts.
Next, we present our detailed implementation for the two most challenging functions, namely \textit{eventVerify()} and \textit{msgRelease()}, respectively.

The implementation for function \textit{eventVerify()} in the proxy contract $C_{proxy}$ is shown as Function~\ref{I1}.
The contract '\textit{RLPReader.sol}' is imported (line 2,6-7) from the \textit{solidity-rlp} library for decoding encoded data input to the \textit{eventVerify()} function.
The functions \textit{verifyBlock()}, \textit{verifyProof()} and \textit{verifyLog()} specified in Algorithm~\ref{A1} are implemented at line 15-20, line 21-34 and line 35-37 in Function~\ref{I1}, respectively.
Specifically, \textit{verifyBlock()} is implemented by first fetching the hash of block specified by the input $\_blockNum$ from the blockchain and verifying that the input $\_rlpBlockData$ has the same hash value (line 15-16), then decoding the input $\_rlpBlockData$ as a list of components (line 17-18) and finally picking out the value for $receiptsRoot$ from the list (line 19-20).
After that, \textit{verifyProof()} is implemented by first verifying the input $\_rlpMerkleRoot$ and decoding $\_rlpMerkleRoot$ to get the hashed value for the branch node on the path of Merkle proof, namely $merkleBranch$ (line 21-25),
then verifying the input $\_rlpMerkleBranch$ and decoding it to get the hashed value for the leaf node $merkleLeaf$ (line 26-30)
and finally verifying the input $\_rlpMerkleLeaf$ and decoding it to get the data for the declared transaction receipt $rlpReceipt$ (line 31-34).
In the third step, \textit{verifyLog()} is implemented by decoding $rlpReceipt$ to get the data for the declared log $rlpLog$ (line 35-37).
Finally, after verifying the equivalence between the log and the specified event (line 38), $msgRelease()$ is invoked to release the reserved transaction (line 40) and the service fee is sent to the executor (line 41).

The function \textit{msgRelease()} could be implemented as one of three functions shown in Function~\ref{I2}, which depends on the type of reserved transaction.
Specifically, \textit{msgRelease()} for releasing a reserved fund transfer transaction, a reserved function invocation transaction or a reserved contract creation transaction could be implemented as line 3-7, line 9-16 or line 18-25 in Function~\ref{I2}, respectively. 

\subsection{Evaluation}

Similar to recent work on blockchain-based platforms and protocols~\cite{dziembowski2019perun,das2018yoda}, the key focus of our evaluation is on measuring gas consumption, namely the amount of transaction fees spent in the protocol. This is due to the fact that the execution complexity in Ethereum is measured via gas consumption.
It is worth noting that, in Ethereum, the amount of gas that a single transaction may spend is bounded by a system parameter and hence, the time overhead of executing functions inside smart contracts is small, usually in the scale of hundreds of milliseconds.

\begin{table}
\begin{center}
\begin{tabular}{|c |c |p{3.0cm}|p{0.9cm}|p{0.6cm}|}
\hline
    \textbf{Phase} & \textbf{Step} & \textbf{Function} & \textbf{Gas} & \textbf{UDS} \\ \hline
    \multirow{3}{*}{\textbf{ET.schedule}}
    & 1 & deploy $C_{proxy}$  & 889764 & \$2.60 \\
    & 2 &  $charge()$ & 21497 & \$0.06 \\
    & 3 & $newService()$ & 45612 & \$0.13 \\ \hline
    \textbf{ET.execute}
    & 4 & $eventVerify()$ & 175674 & \$0.51 \\ \hline
    \textbf{Other}
    &  & $close()$ & 13662 & \$0.04 \\ \hline
\end{tabular}
\end{center}
\vspace{-2mm}
\caption{Key functions and their cost in Gas and USD.}      
\vspace{-6mm}
\label{t4}
\end{table}

In TABLE~\ref{t4}, we list the key functions in the programmed smart contracts that interact with protocol participants during different phases of the protocol and the cost of these functions in both Gas and USD.
The cost in USD was computed through 
$cost(USD)=cost(Gas)*GasToEther*EtherToUSD$,
where $GasToEther$ and $EtherToUSD$ were taken as their mean value during the first half of the year 2019 recorded in \textit{Etherscan}~\cite{etherscan}, 
which are $1.67*10^{-8}$ Ether/Gas and 175 USD/Ether, respectively.
As illustrated by the results, in \texttt{EventWarden}, the completion of a service only requires a user to deploy the proxy contract $C_{proxy}$ (\$2.60), transfer the required amount of ether to $C_{proxy}$ via function $charge()$ (\$0.06) and set up a new service at the hub contract $C_{hub}$ via function $newService()$ (\$0.20) during \textit{ET.schedule}
and an executor to invoke \textit{eventVerify()} (\$0.51) during \textit{ET.execute}, which costs only \$3.30 in total.
The user may also choose to call function $close()$ to shut down the service and get back her ether from the proxy contract, which costs \$0.04.

In summary, our implementation and evaluation demonstrate that \texttt{EventWarden} is effective and is ready-to-use.

\section{Conclusion}
\label{s8}
This paper proposes \texttt{EventWarden}, a decentralized event-driven proxy service for users to outsource any type of transaction in Ethereum-like blockchains.
\texttt{EventWarden} employs a novel combination of smart contracts and blockchain logs.
\texttt{EventWarden} allows a user to create a proxy smart contract that specifies an interested event and also reserves an arbitrary transaction to release.
Upon observing the occurrence of the prescribed event, anyone in the Blockchain network can call the proxy contract to earn the service fee reserved in the contract by proving to the contract that the event has been recorded into blockchain logs, which then automatically triggers the proxy contract to release the reserved transaction.
We show that the reserved transaction can only get released from the proxy contract when the prescribed event has taken place.
We also demonstrate that as long as a single member in the Blockchain network is incentivized by the service fee to call the proxy contract after the prescribed event has taken place, the reserved transaction is guaranteed to get released.
We implement \texttt{EventWarden} over the Ethereum official test network. The results demonstrate that \texttt{EventWarden} is effective and is ready to be used in practice.

\section*{Acknowledgement}
Chao Li acknowledges the partial support by the Fundamental Research Funds for the Central Universities (Grant No. 2019RC038).

\renewcommand\refname{Reference}
\bibliographystyle{IEEEtran}
\urlstyle{same}

\bibliography{main.bib}



\end{document}